\documentclass{llncs}
\usepackage{alltt}
\usepackage{amssymb}
\usepackage[T1]{fontenc}
\usepackage[latin1]{inputenc}
\author{Yves Bertot\thanks{This work was partially supported by ANR contract
Compcert, ANR-05-SSIA-0019.}}
\institute{INRIA Sophia-Méditerranée}
\title{Structural abstract interpretation\\
A formal study using Coq}
\begin{document}
\maketitle
\begin{abstract}
Abstract interpreters are tools to compute approximations for
behaviors of a program.  These approximations can then be used for
optimisation or for error detection.  In this paper,
we show how to describe an abstract interpreter using the type-theory
based theorem prover Coq, using inductive types for syntax and structural
recursive programming for the abstract interpreter's kernel.  The abstract
interpreter can then be proved correct with respect to a Hoare logic
for the programming language.
\end{abstract}
\section{Introduction}
Higher-order logic theorem provers provide a description language that
is powerful enough to describe programming languages.  Inductive types
can be used to describe the language's main data structure (the syntax)
and recursive functions can be used to describe the behavior of instructions
(the semantics).  Recursive functions can also be used to describe  
tools to analyse or modify programs.  In this paper, we will describe
such a collection of recursive function to analyse programs, based
on abstract interpretation \cite{CousotCousot77}.

\subsection{An example of abstract interpretation}
We consider a small programming language with
loop statements and assignments.  Loops are written with the keywords
{\tt while}, {\tt do} and {\tt done}, assignments are written with
{\tt :=}, and several instructions can be grouped together, separating
them with a semi-column.  The instructions grouped using a semi-column
are supposed to be executed in the same order as they are written.
Comments are written after two slashes {\tt //}.

We consider the following simple program:
\begin{verbatim}
x:= 0;             // line 1
While x < 1000 do  // line 2
  x := x + 1       // line 3
done               // line 4
\end{verbatim}
We want to design a tool that is able to gather information about
the value of the variable {\tt x} at each position in the program.
For instance here, we know that after executing the first line,
{\tt x} is always in the interval [0,0]; we know that before
executing the assignment on the third line, {\tt x} is always 
smaller than 10 (because the test {\tt x < 10} was just satisfied).
With a little thinking, we can also guess that {\tt x} increases
as the loop executes, so that we can infer that before the third line,
{\tt x} is always in the interval [0,9].  On the other hand,
after the third line, {\tt x} is always in the interval [1, 10].
Now, if execution exits the loop, we can also infer that the test
{\tt x < 10} failed, so that we know that {\tt x} is larger than or
equal to {\tt 10},
but since it was at best in [0,10] before the test, we can guess that
{\tt x} is exactly 10 after executing the program.  So we can write
the following new program, where the only difference is the information
added in the comments:
\begin{verbatim}
// Nothing is known about x on this line
x := 0;            // 0 <= x <= 0
while x < 10 do
                   // 0 <= x <= 9
  x := x + 1       // 1 <= x <= 10
done
                   // 10 <= x <= 10
\end{verbatim}
We want to produce a tool that performs this analysis and produces
the same kind of information for each line in the program.  Our tool will
do slightly more: first it will also be able to take as input extra
information about variables before entering the program, second it
will produce information about variables after executing the program,
third it will associate an {\em invariant} property to all while
loops in the program.  Such an invariant is a property that is true
{\em before and after} all executions of the loop body (in our example
the loop body is {\tt x := x+1}).
A fourth feature of our tool is that it will be able to detect occasions
when we can be sure that some code is never executed.  In this case,
it will mark the program points that are never reached with a
{\tt false} statement meaning ``when this point of the program is
reached, the {\tt false} statement can be proved (in other words,
this cannot happen)''.

Our tool will also be designed in such a way that it is guaranteed
to terminate in reasonable time.  Such a tool is called a
{\em static} analysis tool, because
the extra information can be obtained without running the program:
in this example, executing the program requires at least a thousand
operations, but our reasoning effort takes less than ten steps.

Tools of this kind are useful, for example to avoid bugs in programs
or as part of efficient compilation techniques.
  For instance, the first mail-spread virus exploited
a programming error known as 
a buffer overflow (an array update was operating outside the memory
allocated for that array), but buffer overflows can be detected if
we know over which interval each variable is likely to range.

\subsection{Formal description and proofs}
Users should be able to trust the information added in programs by
the analysers.  Program analysers are themselves programs and
we can reason about their correctness.  The program analysers
we study in this paper are based on 
abstract interpretation \cite{CousotCousot77} and we use 
the Coq system \cite{Coq,coqart} to reason on its correctness.
The development described in this paper is
available on the net at the following address (there are two versions,
compatible with the latest stable release of Coq ---V8.1pl3--- and with 
the upcoming version ---V8.2).
\begin{center}
\tt http://hal.inria.fr/inria-00329572
\end{center}
This paper has 7 more sections.  Section~2 gives a rough introduction
to the notion of abstract interpretation. Section~3 describes the programming
language that is used as our playground.  The semantics
of this programming language is described using a weakest pre-condition
calculus.  This weakest pre-condition calculus is later used to argue
on the correctness of abstract interpreters.  In particular, abstract
interpretation returns an annotated instruction and an abstract state,
where the abstract state is used as a post-condition and the
annotations in the instruction describe the abstract state at the
corresponding point in the program.  Section~4 describes a first simple
abstract interpreter, where the main ideas around abstractly
interpreting assignments
and sequences are covered, but while loops are not treated.
  In Section~4, we also show that the abstract
interpreter can be formally proved correct.  In Section~5, we address
while loops in more detail and in particular we show how tests can
be handled in abstract interpretation, with applications to dead-code
elimination.  In Section~6, we observe that abstract interpretation is
a general method that can be applied to a variety of abstract domains and
we recapitulate the types, functions, and properties that are expected from each
abstract domain.  In Section~7, we show how the main abstract interpreter
can be instantiated for a domain of intervals, thus making the analysis
presented in the introduction possible.
In Section~8, we give a few concluding remarks.
\section{An intuitive view of abstract interpretation}
Abstract interpretation is a technique for the static analysis of programs.
The objective is to obtain a tool that will take programs as data, perform
some symbolic computation, and return information about all executions
of the input programs.  One important aspect is that this tool should always
terminate (hence the adjective {\em static}).  The tool can then be used
either directly to provide information about properties of variables
in the program (as in the Astree tool \cite{astree05}), or as part of
a compiler, where it can be used to guide optimization.  For instance,
the kind of interval-based analysis that we describe in this paper can
be used to avoid runtime array-bound checking in languages that
impose this kind of discipline like Java.

The central idea of abstract interpretation is to replace the values
normally manipulated in a program by sets of values, in such a way that all
operations still make sense.

For instance, if a program manipulates integer values and performs
additions, we can decide to take an abstract point of
view and only consider whether values are odd or even.  With
respect to addition,
we can still obtain meaningful results, because we know,  for instance,
that adding an even and an odd value returns an odd value.
Thus, we can decide to run programs with values taken in a new type
that contains values {\tt even} and {\tt odd}, with an addition that
respects the following table:
\begin{eqnarray*}
{\tt odd} + {\tt even} &=& {\tt odd}\\
{\tt even} + {\tt odd} &=& {\tt odd}\\
{\tt odd} + {\tt odd} &=& {\tt even}\\
{\tt even} + {\tt even} &=& {\tt even}.
\end{eqnarray*}

When defining abstract interpretation for a given abstract domain, all
operations must be updated accordingly.  The behavior of control instructions
is also modified, because abstract values may not be precise
enough to decide how a given decision should be taken.

For instance, if we know that the abstract value for a variable {\tt x}
is {\tt odd}, then we cannot tell which branch of a conditional statement
of the following form will be taken:
\begin{center}\tt
  if x < 10 then x := 0 else x := 1.
\end{center}
After the execution of this conditional statement, the abstract value
for {\tt x} cannot be {\tt odd} or {\tt even}.
This example also shows that the domain of abstract values must contain
an abstract value that represents the whole set of values, or said differently, 
an abstract value that represents the absence of knowledge.  This
value will be called {\tt top} later in the paper.

There must exist a connection between abstract values and concrete values
for abstract interpretation to work well.  This connection has been studied
since \cite{CousotCousot77} and is known as a Galois connection.  For instance,
if the abstract values are {\tt even}, {\tt odd}, and {\tt top}, and if
we can infer that a value is in \{1,2\}, then correct choices for the
abstract value are {\tt top} or {\tt even}, but obviously the abstract
interpreter will work better if the more precise {\tt even} is chosen.

Formal proofs of correctness
for abstract interpretation were already studied before,
in particular in  \cite{PichardieThesis}.  The approach taken in this
paper is different, in that it follows directly the syntax of a simple
structured programming language, while traditional descriptions are tuned
to studying a control-flow graph language. The main advantage of our approach
is that it supports a very concise description of the abstract interpreter,
with very simple verifications that it is terminating.

\section{The programming language}
In this case study, we work with a very small language containing only
assignments, sequences, and while loops.  The right-hand sides for assignments
are expressions made of numerals, variables, and addition.
The syntax of the programming language is as follows:
\begin{itemize}
\item variable names are noted \(x\), \(y\), \(x_1\), \(x'\), etc.
\item integers are noted \(n\), \(n_1\), \(n'\), etc.
\item Arithmetic expressions are noted 
\(e\), \(e_1\), \(e'\), etc.  For our case study,
these expressions can only take three forms:
\[ e ::= n~ |~ x ~|~ e_1 + e_2 \]
\item boolean expressions are noted \(b\), \(b_1\), \(b'\), etc. For
our case study, these expressions can only take one form:
\[b ::= e_1 < e_2\]
\item instructions are noted \(i\), \(i_1\), \(i'\), etc.  For our
case study, these instructions can only take three forms:
 \[i ::= x {\tt :=} e ~|~ i_1{\tt;} i_2 ~|~ {\tt while}~b~{\tt do}~i~{\tt done}\]
\end{itemize}

For the Coq encoding, we use pre-defined strings for variable names and
integers for the numeric values.  Thus, we use unbounded integers, which
is contrary to usual programming languages, but the question of using
bounded integers or not is irrelevant for the purpose of this
example.

\subsection{Encoding the language}
In our Coq encoding, the description of the various kinds of syntactic
components is given by inductive declarations.
\begin{verbatim}
Require Import String ZArith List.
Open Scope Z_scope.

Inductive aexpr : Type :=
  anum (x:Z) | avar (s:string) | aplus (e1 e2:aexpr).

Inductive bexpr : Type := blt (e1 e2 : aexpr).

Inductive instr : Type := 
  assign (x:string)(e:expr)
| seq (i1 i2:instr)
| while (b:bexpr)(i:instr).
\end{verbatim}
The first two lines instruct Coq to load pre-defined libraries and to
tune the parsing mechanism so that arithmetic formulas will be
understood as formulas concerning integers by default.

The definition for {\tt aexpr} states that expressions can only have
the three forms {\tt anum}, {\tt avar}, and {\tt aplus}, it also expresses
that the names {\tt anum}, {\tt avar}, and {\tt aplus} can be used
as function of type, {\tt Z -> aexpr}, {\tt string -> aexpr}, and
{\tt aexpr -> aexpr -> aexpr}, respectively.  The definition of {\tt aexpr}
as an inductive type also implies that we can write recursive functions
on this type.  For instance, we will use the following function to evaluate
an arithmetic expression, given a {\em valuation} function {\tt g}, which
maps every variable name to an integer value.
\begin{verbatim}
Fixpoint af (g:string->Z)(e:aexpr) : Z :=
  match e with
    anum n => n
  | avar x => g x
  | aplus e1 e2 => af g e1 + af g e2
  end.
\end{verbatim}
This function is defined by pattern-matching.  There is one pattern
for each possible form of arithmetic expression.  The third line
indicates that when the input {\tt e} has the form {\tt anum n}, then
the value {\tt n} is the result.  The fourth line indicates that
when the input has the form {\tt avar x}, then the value is obtained
by applying the function {\tt g} to {\tt x}.  The fifth line describes
the computation that is done when the expression is an addition.
There are two recursive calls to the function {\tt af}
in the
expression returned for the addition pattern.  The recursive calls
are made on direct subterms of the initial instruction, this is known
as {\em structural recursion} and guarantees that the recursive function
will terminate on all inputs.

A similar function {\tt bf} is defined to describe the boolean
value of a boolean expression.

\subsection{The semantics of the programming language}
To describe the semantics of the programming language, we simply give
a {\sl weakest pre-condition calculus} \cite{Dijkstra76}.
We describe the conditions that are
necessary to ensure that a given logical property is satisfied at the
end of the execution of an instruction, when this execution terminates.  This
weakest pre-condition calculus is defined as a pair of functions whose
input is an instruction annotated
with logical information at various points in the instruction.  The output
of the first function call {\tt pc}
is a condition that should be satisfied by the variables at the beginning
of the execution (this is the {\sl pre-condition} and it should be as easy
to satisfy as possible, hence the adjective {\sl weakest});
the output of the second function, called {\tt vc},
is a collection of logical statements.  When these statements are
valid, we know that every execution starting from a state that satisfies
the pre-condition will make the logical annotation satisfied at every
point in the program and make the post-condition satisfied if the execution
terminates.

\subsubsection{annotating programs}
We need to define a new data-type for instructions
annotated with assertions at various locations.  Each assertion
is a quantifier-free logical formula where the variables of the
program can occur.  The intended meaning is that the formula
is guaranteed to hold for every execution of the program that is
consistent with the initial assertion.

The syntax for assertions is described as follows:
\begin{verbatim}
Inductive assert : Type :=
  pred (p:string)(l:list aexpr)
| a_b (b:bexpr)
| a_conj (a1 a2:assert)
| a_not (a: assert)
| a_true
| a_false.
\end{verbatim}
This definition states that assertions can have six forms:
the first form represents the application of a predicate to
an arbitrary list of arithmetic expressions, the second represents
a boolean test: this assertion holds when the
boolean test evaluates to {\tt true}, the third form is the conjunction
of two assertions, the fourth form is the negation of an assertion,
the fifth and sixth forms give two constant assertions, which are
always and never satisfied, respectively.  In a minimal description
of a weakest pre-condition calculus, as in \cite{Bertot08a}, the last
two constants are not necessary, but they will be useful in our
description of the abstract interpreter.

Logical annotations play a central role in our case study, because
the result of abstract interpretation will be to add information about
each point in the program: this new information will be described by
assertions.

To consider whether an assertion holds, we need to know what meaning
is attached to each predicate name and what value is attached to
each variable name.  We suppose the meaning of predicates is
given by a function {\tt m} that maps predicate names and list of
integers to propositional values and the value of variables is given by
a valuation as in the function {\tt af} given above.  Given such
a meaning for predicates and such a valuation function for variables,
we describe the computation of the property associated to an assertion
as follows:
\begin{verbatim}
Fixpoint ia (m:string->list Z->Prop)(g:string->Z)
   (a:assert) : Prop :=
  match a with
    pred s l => m s (map (af g) l)
  | a_b b => bf g b = true
  | a_conj a1 a2 => (ia m g a1) /\ (ia m g a2)
  | a_not a => not (ia m g a)
  | a_true => True
  | a_false => False
  end.
\end{verbatim}
The type of this function exhibits a specificity of type theory-based
theorem proving: propositions are described by {\em types}.  The
Coq system also provides a type of types, named {\tt Prop}, whose
elements are the types that are intended to be used as propositions.  Each
of these types contains the proofs of the proposition they represent.
This is known as the {\em Curry-Howard isomorphism}.  For instance,
the propositions that are unprovable are represented
by empty types.  Here, assertions are data, their
interpretation as propositions are types, which belongs to the {\tt Prop}
type.  More details about this description of propositions as types
is given in another article on type theory in the same volume.

Annotated instructions are in a new data-type, named
{\tt a\_instr}, which is very close to the {\tt instr} data-type.
The two modifications are as follows: first an extra operator {\tt pre}
is added to make it possible to attach assertions to any instruction,
second {\tt while}
loops are mandatorily annotated wih an {\em invariant} assertion.
In concrete syntax, we will write {\tt \{ \(a\) \} \(i\)} for the instruction
\(i\) carrying the assertion \(a\) (noted {\tt pre a i} in the Coq encoding).
\begin{verbatim}
Inductive a_instr : Type :=
  pre (a:assert)(i:a_instr)
| a_assign (x:string)(e:aexpr)
| a_seq (i1 i2:a_instr)
| a_while (b:bexpr)(a:assert)(i:a_instr).
\end{verbatim}
\subsubsection{Reasoning on assertions}
We can reason on annotated programs, because there are logical reasons
for programs to be consistent with assertions.  The idea is to
compute a collection of logical formulas associated to an annotated program
and a final logical formula, the {\em post-condition}.  When this
collection of formulas holds, there exists an other logical formula,
the {\em pre-condition} whose satisfiability before executing the program
is enough to guarantee that the post-condition holds after executing the
program.

  Annotations added to an instruction (with
the help of the {\tt pre} construct) must be understood as formulas that
hold just before executing the annotated instruction.  Assertions
added to {\tt while} loops must be understood as {\em invariants}, they
are meant to hold at the beginning and the end every time the
inner part of the while loop is executed.

When assertions are present in the annotated
instruction, they are taken for granted.
For instance, when the instruction is {\tt \{x = 3\} x := x + 1 }, the
computed pre-condition is {\tt x = 3}, whatever the post-condition is.

When the instruction is a plain assignment, one can find the pre-condition
by substituting the assigned variable with the assigned expression
in the post-condition.  
For instance, when the post condition is {\tt x = 4} and the instruction
is the assignement {\tt x := x + 1}, it suffices that the pre-condition
{\tt x + 1 = 4}  is satisfied before executing the assignment to ensure that
the post-condition is satisfied after executing it.

When the annotated instruction is a while loop, the pre-condition simply is the
invariant for this while loop.  When the annotated instruction is a sequence of
two instructions, the pre-condition is the pre-condition computed for
the first of the two instructions, but using the pre-condition of
the second instruction as the post-condition for the first instruction.

\subsubsection{Coq encoding for pre-condition computation}
To encode this pre-condition function in Coq, we need to describe
functions that perform the substitution of a variable with an arithmetic
expression in arithmetic expressions, boolean expressions, 
and assertions.  These substitution functions
are given as follows:
\begin{verbatim}
Fixpoint asubst (x:string) (s:aexpr) (e:aexpr) : aexpr :=
  match e with
    anum n => anum n
  | avar x1 => if string_dec x x1 then s else e
  | aplus e1 e2 => aplus (asubst x s e1) (asubst x s e2)
  end.

Definition bsubst (x:string) (s:aexpr) (b:bexpr) : bexpr :=
  match b with
    blt e1 e2 => blt (asubst x s e1) (asubst x s e2)
  end.

Fixpoint subst (x:string) (s:aexpr) (a:assert) : assert :=
  match a with
    pred p l => pred p (map (asubst x s) l)
  | a_b b => a_b (bsubst x s b)
  | a_conj a1 a2 => a_conj (subst x s a1) (subst x s a2)
  | a_not a => a_not (subst x s a)
  | any => any
  end.
\end{verbatim}
In the definition of {\tt asubst}, the function {\tt string\_dec} compares
two strings for equality.  The value returned by this function can be used
in an {\tt if-then-else} construct, but it is not a boolean value (more
detail can be found in \cite{coqart}).  The rest of the code is just a plain
traversal of the structure of expressions and assertions.  Note also
that the last pattern-matching rule in {\tt subst} is used for both
{\tt a\_true} and {\tt a\_false}.

Once we know how to substitute a variable with an expression, we
can easily describe the computation of the pre-condition for
an annotated instruction and a post-condition.
This is given by the following simple recursive
procedure:
\begin{verbatim}
Fixpoint pc (i:a_instr) (post : assert) : assert :=
  match i with
    pre a i => a
  | a_assign x e => subst x e post
  | a_seq i1 i2 => pc i1 (pc i2 post)
  | a_while b a i => a
  end.
\end{verbatim}

\subsubsection{A verification condition generator}
When it receives an instruction carrying an annotation,
the function {\tt pc} simply returns the annotation.  In this sense, the
pre-condition function takes the annotation for granted.  To make sure
that an instruction is consistent with its pre-condition, we need to check that
the assertion really is strong enough to ensure the post-condition.

For instance, when the post-condition is {\tt x < 10} and the instruction
is the annotated assigment {\tt \{ x = 2 \} x := x + 1}, satisfying
{\tt x = 2} before the assignment is enough to ensure that the post-condition
is satisfied.  On the other hand, if the annotated instruction was
{\tt \{x < 10 \} x := x + 1}, there would be a problem because there are
cases where {\tt x < 10} holds before executing the assignment and {\tt x < 10}
does not hold after.

In fact, for assigments that are not annotated with assertions, the function
{\tt pc} computes the best formula, the {\em weakest pre-condition}.  Thus,
in presence of an annotation, it suffices to verify that the annotation
does imply the weakest pre-condition.  We are now going to describe a function
that collects all the verifications that need to be done.
More precisely, the new function will compute conditions that are sufficient
to ensure that the pre-condition from the previous section is strong
enough to guarantee that the post-condition holds after executing the program,
when the program terminates.

The verification that an annotated instruction is consistent with a
post-condition thus returns a sequence of implications between
assertions.  When all these implications are logically valid, there is
a guarantee that satisfying the pre-condition before executing the
instruction is enough to ensure that the post-condition will also be satisfied
after executing the instruction.  This guarantee is proved formally in
\cite{Bertot08a}.

When the instruction is a plain assignment without annotation, there is
no need to verify any implication because the computed pre-condition is
already good enough.  When the instruction is an annotated instruction
{\tt \{ \(A\) \} \(i\)} and the post-condition is \(P\), we can first
compute the pre-condition \(P'\) and a list of implications \(l\) for
the instruction \(i\) and the post-condition \(P\).  We then only need
to add \( A \Rightarrow P'\) to \(l\) to get the list of conditions
for the whole instruction.

For instance, when the post-condition is {\tt x=3} and the instruction
is the assignment {\tt x := x+1}, the pre-condition computed by {\tt pc} is
{\tt x + 1 = 3} and this is obviously good enough for the post-condition
to be satisfied.  On the other hand, when the instruction is an annotated
instruction, {\tt \{\(P\)\} x := x+1}, we need to verify that 
\(P \Rightarrow {\tt x + 1 = 3}\) holds.  

If we look again at the first example in this section,
concerning an instruction {\tt \{x < 10\} x := x+1} and a post-condition
{\tt x < 10}, there is a problem, because a value of 9 satisfies the
pre-condition, but execution leads to a value of 10, which does not
satisfy the post-condition  The condition
generator constructs a condition of the form
{\tt x < 10 \(\Rightarrow\) x + 1 < 10}.  The fact that this logical formula
is actually unprovable relates to the fact that the triplet composed of
the pre-condition, the assignment, and the post-condition is actually
inconsistent.

When the instruction is a sequence of two instructions {\tt \(i_1\);\(i_2\)}
and the post-condition is \(P\),
we need to compute lists of conditions for both sub-components \(i_1\) and
\(i_2\).  The list of conditions for \(i_2\) is computed for the post-condition
for the whole construct \(P\), while the list of conditions of \(i_1\) is
computed taking as post-condition the pre-condition of \(i_2\) for \(P\).
This is consistent with the intuitive explanation that it suffices
that the pre-condition for an instruction holds to ensure that the 
post-condition will hold after executing that instruction.  If we want \(P\)
to hold after executing \(i_2\), we need the pre-condition of \(i_2\) for
\(P\) to be satisfied and it is the responsibility of the instruction \(i_1\)
to guarantee this.  Thus, the conditions for \(i_1\) can be computed
with this assertion as a post-condition.

When the instruction is a while loop, of the form
{\tt while \(b\) do \{ \(A\) \} \(i\) done} we must remember that the assertion
\(A\) should be an invariant during the loop execution.
This is expressed by requiring that \(A\) is satisfied before executing
\(i\) should be enough to guarantee that \(A\) is also satisfied after
executing \(i\).  However, this is needed only in the cases where the
loop test \(b\) is also satisfied, because when \(b\) is not satisfied the
inner instruction of the while loop is not executed.  At the end of the
execution, we can use the information that the invariant \(A\) is satisfied
and the information that we know the loop has been executed because the test
eventually failed.  The program is consistent when these two logical 
properties are enough to imply the initial post-condition \(P\).
Thus, we must first
compute the pre-condition \(A'\) for the inner instruction \(i\) and the
post-condition \(A\), compute the list of conditions for \(i\) with
\(A\) as post-condition, add the condition \(A\wedge b\Rightarrow A'\),
and add the condition \(A\wedge \neg b \Rightarrow P\).

\subsubsection{Coq encoding of the verification condition generator}
The verification conditions always are implications.  We provide a new
data-type for these implications:
\begin{verbatim}
Inductive cond : Type := imp (a1 a2:assert).
\end{verbatim}
The computation of verification conditions is then simply described
as a plain recursive function, which follows the structure of 
annotated instructions.
\begin{verbatim}
Fixpoint vc (i:a_instr)(post : assert) : list cond :=
  match i with
    pre a i => (imp a (pc i post))::vc i post
  | a_assign _ _ => nil
  | a_seq i1 i2 => vc i1 (pc i2 post)++vc i2 post
  | a_while b a i =>
    (imp (a_conj a (a_b b)) (pc i a))::
    (imp (a_conj a (a_not (a_b b))) post)::
    vc i a
  end.
\end{verbatim}

Describing the semantics of programming language using a verification
condition generator is not the only approach that can be used to describe
the language.  In fact, this approach is partial, because it describes
properties of inputs and outputs when instruction execution terminates, but
it gives no information about termination.  More precise descriptions can
be given using operational or denotational semantics and the consistency
of this verification condition generator with such a complete semantics
can also be verified formally.  This is done in \cite{Bertot08a}, but it
is not the purpose of this article.

When reasoning about the correctness of a given annotated instruction,
we can use the function {\tt vc} to obtain a list of conditions.  It is
then necessary to reason on the validity of this list of conditions.
What we want to verify is that the implications hold for every possible
instantiation of the program variables.  This is described by the following
function.
\begin{verbatim}
Fixpoint valid (m:string->list Z ->Prop) (l:list cond) : Prop :=
  match l with
    nil => True
  | c::tl =>
    (let (a1, a2) := c in forall g, ia m g a1 -> ia m g a2)
      /\ valid m tl
  end.
\end{verbatim}
An annotated program \(i\) is consistent with a given post-condition 
\(p\) when the property {\tt valid (vc \(i\) \(p\))} holds.  This means
that the post-condition is guaranteed to hold after executing the
instruction if the computed pre-condition was satisfied before the
execution and the execution of the instruction terminates.

\subsection{A monotonicity property}
In our study of an abstract interpreter, we will use a property of the
condition generator.

\begin{theorem} For every annotated instruction \(i\),
if \(p_1\) and \(p_2\) are two post-conditions such
that \(p_1\) is stronger than \(p_2\), if the pre-condition for \(i\) and
\(p_1\)
is satisfied and all the verification conditions for \(i\) and 
the post-condition
\(p_1\) are valid, then the pre-condition for \(i\) and \(p_2\)
is also satisfied and 
the verification conditions for \(i\) and \(p_2\) are also valid.
\end{theorem}
\begin{proof}
This proof is done in the context of a given mapping from predicate
names to actual predicates, \(m\).  The
 property is proved by induction on the structure of the instruction
\(i\).
The statement \(p_1\) is stronger than \(p_2\)  when the
implication \(p_1\Rightarrow p_2\) is valid.  In other words, for
every assignment of variables \(g\), the logical value of \(p_1\) implies
the logical value of \(p_2\).

If the instruction is an assignment,
we can rely on a lemma: the value of any assertion
\({\tt subst}~x~e~p\) in any valuation \(g\) is equal to the value of
the assertion \(p\) in the valuation \(g'\) that is equal to \(g\) on
every variable but \(x\), for which it returns the value of \(e\) in
the valuation \(g\).  Thus, the precondition for the assignment
{\tt \(x\) := \(e\)} for \(p_i\) is \({\tt subst}~x~e~p_i\) and the
the validity of \({\tt subst}~x~e~p_1 \Rightarrow {\tt subst}~x~e~p_2\) simply
is an instance of the validity of \(p_1\Rightarrow p_2\), which is given
by hypothesis.  Also, when the instruction is an assignment, there is
no generated verification condition and the second part of the statement
holds.

If the instruction is a sequence \( i_1 {\tt;} i_2\), then we know
by induction hypothesis that the pre-condition \(p'_1\) 
for \(i_2\) and \(p_1\) is
stronger than the pre-condition \(p'_2\) for \(i_2\) and \(p_2\) and all the
verification conditions for that part are valid; we can use an induction
hypothesis again to obtain that the pre-condition for \(i_1\) and \(p'_1\) is
stronger than the pre-condition for \(i_1\) and \(p'_2\), and the corresponding
verification conditions are all valid.  The last
two pre-conditions are the ones we need to compare, and the whole
set of verification conditions is the union of the sets which we know
are valid.

If the instruction is an annotated instruction \({\tt\{}a{\tt\}}i\),
the two pre-conditions for \(p_2\) and \(p_1\) alre always \(a\), so
the first part of the statement trivially holds.  Moreover, we
know by induction hypothesis that the pre-condition \(p'_1\)
for \(i\) and \(p_1\) is stronger that the pre-condition \(p'_2\) for
\(i\) and \(p_2\).  The verification conditions for the whole instruction
and \(p_1\) (resp. \(p_2\)) are the same as for the sub-instruction,
with the condition  \(a \Rightarrow p'_1\) (resp. \(a\Rightarrow p'_2\))
added.  By hypothesis, \(a\Rightarrow p'_1\) holds, by induction
hypothesis \(p'_1\Rightarrow p'_2\), we can thus deduce that 
\(a\Rightarrow p'_2\) holds.

If the instruction is a loop \({\tt while}~b~{\tt do} \{a\}~i~{\tt done}\),
most verification conditions and generated
pre-conditions only depend on the loop invariant.  The only thing
that we need to check is the verification condition containing the
invariant, the negation of the test and the post-condition.  By hypothesis,
\(a\wedge \neg b \Rightarrow p_1\) and \(p_1\Rightarrow p_2\) are valid.
By transitivity of implication we obtain \(a\wedge\neg b\Rightarrow p_2\)
easily.
\end{proof}
In Coq, we first prove a lemma that expresses that
the satisfiability of an assertion {\tt a} where a variable {\tt x} is
substituted with an arithmetic expression {\tt e'} for a valuation {\tt g}
is the same as the satisfiability of the assertion {\tt a} without
substitution, but for a valuation that maps {\tt x} to the value of {\tt e'}
in {\tt g} and coincides with {\tt g} for all other variables.
\begin{verbatim}
Lemma subst_sound :
  forall m g a x e',
    ia m g (subst x e' a) =
    ia m (fun y => if string_dec x y then af g e' else g y) a.
\end{verbatim}
This lemma requires similar lemmas for arithmetic expressions, boolean
expressions, and lists of expressions.  All are proved by induction on
the structure of expressions.

\subsubsection{An example proof for substitution}  For instance, the statement
for the substitution in arithmetic expressions is as follows:
\begin{verbatim}
Lemma subst_sound_a :
  forall g e x e',
   af g (asubst x e' e) =
   af (fun y => if string_dec x y then af g e' else g y) e.
\end{verbatim}
The proof can be done in Coq by an induction on
the expression {\tt e}.  This leads the system to generate
three cases, corresponding to the three constructors of the {\tt aexpr}
type.  The combined tactic we use is as follows:
\begin{verbatim}
intros g e x e'; induction e; simpl; auto.
\end{verbatim}
The tactic {\tt induction e} generates three goals and
the tactics {\tt simpl} and {\tt auto} are applied to all of them.
One of the cases is the case for
the {\tt anum} constructor, where
both instances of the {\tt af} function compute to the
value carried by the constructor, thus {\tt simpl} forces the computation
and leads to an equality where both sides are equal.  In this case, {\tt auto}
solves the goal.  Only the other two goals remain.

The first other goal is concerned with the {\tt avar} construct.  In this
case the expression has the form {\tt avar s} and the expression
{\tt subst x e' (avar s)} is transformed into the following expression
by the {\tt simpl} tactic.
\begin{center}
\tt if string\_dec x s then e' else (avar s)
\end{center}
For this case, the system displays a goal that has the following shape:
\begin{verbatim}
  g : string -> Z
  s : string
  x : string
  e' : aexpr
  ============================
   af g (if string_dec x s then e' else avar s) =
   (if string_dec x s then af g e' else g s)
\end{verbatim}
In Coq goals, the information that appears above the horizontal bar
is data that is known to exist, the information below the horizontal
bar is the expression that we need to prove.  Here the information
that is known only corresponds to typing information.

We need to reason by cases on
the values of the expression {\tt string\_dec x s}.  The tactic
{\tt case ...} is used for this purposes.  It generate two goals,
one corresponding to the case where {\tt string\_dec x s} has
 an affimative value and one corresponding to the case where
{\tt string\_dec x s} has a negative value.  In each the goal,
the {\tt if-then-else} constructs are reduced accordingly.
In the goal where {\tt string\_dec x s}  is affirmative,
both sides of the equality reduce
to {\tt af g e'}; in the other goal,
both sides of the equality reduce to {\tt g x}.  Thus in both cases,
the proof becomes easy.  This reasoning step
is easily expressed with the following combined tactic:
\begin{verbatim}
case (string_dec x s); auto.
\end{verbatim}
There only remains a goal for the last possible form
of arithmetic expression, {\tt aplus e1 e2}.  The induction tactic 
provides {\em induction hypotheses} stating that the property we
want to prove already holds for {\tt e1} and {\tt e2}.  After
symbolic computation of the functions {\tt af} and {\tt asubst}, as
performed by the {\tt simpl} tactic, the goal has the following shape:
\begin{verbatim}
  ...
  IHe1 : af g (asubst x e' e1) =
         af (fun y : string =>
              if string_dec x y then af g e' else g y) e1
  IHe2 : af g (asubst x e' e2) =
         af (fun y : string =>
              if string_dec x y then af g e' else g y) e2
  ============================
   af g (asubst x e' e1) + af g (asubst x e' e2) =
   af (fun y : string =>
        if string_dec x y then af g e' else g y) e1 +
   af (fun y : string =>
        if string_dec x y then af g e' else g y) e2
\end{verbatim}
This proof can be finished by rewriting with the two equalities
named {\tt IHe1} and {\tt IHe2} and then recognizing that
both sides of the equality are the same, as required by the following tactics.
\begin{verbatim}
rewrite IHe1, IHe2; auto.
Qed.
\end{verbatim}

We can now turn our attention to the main result,
which is then expressed as the following statement:
\begin{verbatim}
Lemma vc_monotonic :
  forall m i p1 p2, (forall g, ia m g p1 -> ia m g p2) ->
   valid m (vc i p1) ->
   valid m (vc i p2) /\
   (forall g, ia m g (pc i p1) -> ia m g (pc i p2)).
\end{verbatim}
To express that this proof is done by induction on the structure of
instructions, the first tactic sent to the proof system has the form:
\begin{verbatim}
intros m; induction i; intros p1 p2 p1p2 vc1.
\end{verbatim}
The proof then has four cases, which are solved in about 10 lines of
proof script.

\section{A first simple abstract interpreter}
We shall now define two abstract interpreters, which run instructions
symbolically, updating an abstract state at each step.  The abstract
state is then transformed into a logical expression which is added
to the instructions, thus producing an annotated instruction.  The
abstract state is also returned at the end of execution, in one of two
forms.  In the first simple abstract interpreter, the final abstract
state is simply returned.  In the seccond abstract interpreter, only
an optional abstract state will be returned, a {\tt None} value being
used when the abstract interpreter can detect that the program can
never terminate: the second abstract interpreter will also perform
dead code detection.

For example, if we give our abstract interpreter an input state stating
that {\tt x} is even and {\tt y} is odd and the instruction
{\tt x:= x+y; y:=y+1}, the resulting value will be:
\begin{verbatim}
({even x /\ odd y} x:=x+y; {odd x /\ odd y} y:= y+1, 
  (x, odd)::(y,even)::nil)
\end{verbatim}

We suppose there exists a data-type \(A\) whose elements will
represent abstract values on which instructions are supposed to
compute.  For instance, the data-type {\tt A} could be the type
containing three values {\tt even}, {\tt odd}, and {\tt top}.
Another traditional example of abstract data-type is the type
of intervals, that are either of the form \([m,n]\),
with \(m\leq n\), \([-\infty,n]\), \([m,+\infty]\), or \([-\infty,+\infty]\).

The data-type of abstract values should come with a few elements and functions,
which we will describe progresssively.
\subsection{Using Galois connections}
Abstract values represent specific sets of concrete values.  There
is a natural order on sets : set inclusion.  Similarly, we
can consider an order on abstract values, which mimics the order
between the sets they represent.  The traditional
approach to describe this correspondance between the order on
sets of values and the order on abstract values is to consider that
the type of abstract values is given with a pair of functions \(\alpha\)
and \(\gamma\), where \(\alpha : {\cal P}(\mathbb{Z})\rightarrow A\) and
\(\gamma : A \rightarrow {\cal P}(\mathbb{Z})\).  The function \(\gamma\)
maps any abstract value to the set of concrete values it represents.
The function
\(\alpha\) maps any set of concrete values to the smallest abstract
value whose interpretation as a set contains the input.  Written
in a mathematical formula where \(\sqsubseteq\) denotes the order on
abstract values, the two functions and the orders on
sets of concrete values and on abstract values are related by the following
statement:
\[\forall a\in A, \forall b \in{\cal P}(\mathbb{Z}). b \subset \gamma(a)
\Leftrightarrow \alpha(b) \sqsubseteq a.\]
When the functions \(\alpha\) and \(\gamma\) are given with this property,
one says that there is a {\em Galois connection}.

In our study of abstract interpretation, the functions \(\alpha\) and
\(\gamma\) do not appear explicitly.  In a sense, \(\gamma\) will
be represented by a function {\tt to\_pred} mapping abstract values
to assertions depending on arithmetic expressions.  However, it is useful
to keep these functions in mind when trying to figure out what properties
are expected for the various components of our abstract interpreters, as
we will see in the next section.

\subsection{Abstract evaluation of arithmetic expressions}
Arithmetic expressions contain integer constants and additions, neither of
which are concerned with the data-type of abstract values.  To be able
to associate an abstract value to an arithmetic expression, we need
to find ways to establish a correspondance between concrete values
and abstract values.  This is done by supposing the existence of two
functions and a constant, which are the first three values axiomatized
for the data-type of abstract values (but there will be more later):
\begin{itemize}
\item {\tt from\_Z : Z -> A}, this
is used to associate a relevant abstract value
to any concrete value,
\item {\tt a\_add : A -> A -> A}, this is used to add two abstract values,
\item {\tt top : A}, this is used to represent the abstract value that
carries no information.
\end{itemize}

In terms of Galois connections, the function {\tt from\_Z} corresponds
to the function \(\alpha\), when applied to singletons.  The function
{\tt a\_add} must be designed in such a way that the following property
is satisfied:
\[\forall v_1~v_2, \{x + y | x\in(\gamma(v_1), y\in(\gamma(v_2))\}
\subset \gamma({\tt a\_add}~v_1~v_2).\]
With this constraint, a function that maps any pairs of abstract values
to {\tt top} would be acceptable, however it would be useless.  It is
better if {\tt a\_add \(v_1\) \(v_2\)} is the least satisfactory 
abstract value such that the above property is satisfied.

The value {\tt top} is the maximal element of \(A\), the
image of the whole \(\mathbb{Z}\) by the function \(\alpha\).

\subsection{Handling abstract states}
\label{lookup-def}
When computing the value of a variable, we suppose that this value is given
by looking up in a state, which actually is a list of pairs of variables
and abstract values.
\begin{verbatim}
Definition state := list(string*A).

Fixpoint lookup (s:state) (x:string) : A :=
  match s with
    nil => top
  | (y,v)::tl => if string_dec x y then v else lookup tl x
  end.
\end{verbatim}
As we see in the definition of {\tt lookup}, when a value is not
defined in a state, the function behaves
as if it was defined with {\tt top} as abstract value.  The computation
of abstract values for arithmetic expressions is then described by
the following function.

\begin{verbatim}
Fixpoint a_af (s:state)(e:aexpr) : A :=
  match e with
    avar x => lookup s x
  | anum n => from_Z n
  | aplus e1 e2 => a_add (a_af s e1) (a_af s e2)
  end.
\end{verbatim}

When executing assignments abstractly, we are also supposed to
modify the state.  If the
state contained no previous information about the assigned variable,
a new pair is created.  Otherwise, the first existing pair must be updated.
This is done with the following function.

\begin{verbatim}
Fixpoint a_upd(x:string)(v:A)(l:state) : state :=
  match l with
    nil => (x,v)::nil
  | (y,v')::tl =>
    if string_dec x y then (y, v)::tl else (y,v')::a_upd x v tl
  end.
\end{verbatim}
Later in this paper, we define a function that generates assertions
from states.  For this purpose, it is better to update by modifying
existing pairs of a variable and a value rather than just inserting
the new pair in front.
\subsection{The interpreter's main function}
When computing abstract interpretation on instructions we want to produce
a final abstract state and an annotated instruction.  We will need a way
to transform an abstract value into an assertion.  This is given by a function
with the following type:
\begin{itemize}
\item {\tt to\_pred : A -> aexpr -> assert} this is used to express that
that the value of the arithmetic expression in a given valuation will
belong to the set of concrete values represented by the given abstract
value.  So {\tt to\_pred} is axiomatized in the same sense as
{\tt from\_Z}, {\tt a\_add}, {\tt top}.
\end{itemize}
Relying on the existence of {\tt to\_pred},
we can define a function that maps states to
assertions:
\begin{verbatim}
Fixpoint s_to_a (s:state) : assert :=
  match s with
    nil => a_true
  | (x,a)::tl => a_conj (to_pred a (avar x)) (s_to_a tl)
  end.
\end{verbatim}
This function is implemented in a manner that all pairs present in
the state are transformed into assertions.  For this reason, it is
important that {\tt a\_upd} works by modifying existing pairs rather
than hiding them.

Our first simple abstract interpreter only implements
a trivial behavior for while loops.  Basically, this says that no information
can be gathered for while loops (the result is {\tt nil}, and the while loop's
invariant is also {\tt nil}).
\begin{verbatim}
Fixpoint ab1 (i:instr)(s:state) : a_instr*state :=
  match i with
    assign x e =>
    (pre (s_to_a s) (a_assign x e), a_upd x (a_af s e) s)
  | seq i1 i2 => 
    let (a_i1, s1) := ab1 i1 s in
    let (a_i2, s2) := ab1 i2 s1 in
      (a_seq a_i1 a_i2, s2)
  | while b i =>
    let (a_i, _) := ab1 i nil in
        (a_while b (s_to_a nil) a_i, nil)
  end.
\end{verbatim}
In this function, we see that the abstract interpretation of sequences
is simply described as composing the effect on states and recombining
the instruction obtained from each component of the sequence.
\subsection{Expected properties for abstract values}
To prove the correctness of the abstract interpreter, we 
need to know that the various functions and values provided around the
type {\tt A} satisfy a collection of properties.  These are gathered as
a set of hypotheses.

One value that we have not talked about yet is the mapping from
predicate names to actual predicates on integers, which is necessary
to interpret the assertions generated by {\tt to\_pred}.  This is given
axiomatically, like {\tt top} and the others:
\begin{itemize}
\item {\tt m : string -> list Z -> Prop}, maps all predicate names
used in {\tt to\_pred} to actual predicates on integers.
\end{itemize}

The first hypothesis expresses that {\tt top} brings no information.
\begin{verbatim}
Hypothesis top_sem : forall e,  (to_pred top e) = a_true.
\end{verbatim}

The next two hypotheses express that
the predicates associated to each abstract value
are {\em parametric} with respect to the arithmetic expression they receive.
Their truth does not depend on the exact shape of the expressions, but only
on the concrete value such an arithmetic expression may take in the current
valuation.  Similarly, substitution basically affects the arithmetic expression
part of the predicate, not the part that depends on the abstract value.
\begin{verbatim}
Hypothesis to_pred_sem :
  forall g v e, ia m g (to_pred v e) =
     ia m g (to_pred v (anum (af g e))).
Hypothesis subst_to_pred :
  forall v x e e', subst x e' (to_pred v e) =
     to_pred v (asubst x e' e).
\end{verbatim}

For instance, if the abstract values are intervals, it is natural that
the {\tt to\_pred} function will map an abstract value {\tt [3,10]} and
an arithmetic expression {\tt e} to an assertion {\tt between(3, e, 10)}.
When evaluating this assertion with respect to a given valuation {\tt g},
the integers {\tt 3} and {\tt 10} will not be affected by {\tt g}.  Similarly,
substitution will not affect these integers.

The last two hypotheses express that the interpretation of the
associated predicates for abstract values obtained through 
{\tt from\_Z} and {\tt a\_add} are consistent with the concrete values
computed for immediate integers and additions.  The hypothesis
{\tt from\_Z\_sem} actually establishes the correspondence between
{\tt from\_Z} and the abstraction function \(\alpha\) of a Galois
connection.  The hypothesis {\tt a\_add\_sem} expresses the
condition which we described informally when introducing the function
{\tt a\_add\_sem}.
\begin{verbatim}
Hypothesis from_Z_sem :
  forall g x, ia m g (to_pred (from_Z x) (anum x)).
Hypothesis a_add_sem : forall g v1 v2 x1 x2,
  ia m g (to_pred v1 (anum x1)) ->
  ia m g (to_pred v2 (anum x2)) ->
  ia m g (to_pred (a_add v1 v2) (anum (x1+x2))).
\end{verbatim}

\subsection{Avoiding duplicates in states}
The way {\tt s\_to\_a} and {\tt a\_upd} are defined is not consistent:
{\tt s\_to\_a} maps every pair occuring in a state to an assertion
fragment, while {\tt a\_upd} only modifies the first pair occuring
in the state.

For instance, when the abstract interpretation computes with intervals,
{\tt s} is {\tt ("x", [1,1])::("x",[1,1])::nil}, and the
instruction is {\tt x := x + 1}, the resulting state is
{\tt ("x",[2,2])::("x",[1,1])::nil} and the resulting annotated instruction
is {\tt \{ \(1\leq {\tt x} \leq 1\wedge 1\leq {\tt x} \leq 1\)\} x:= x+1}.
The post-condition corresponding to the resulting state is
\(2\leq {\tt x} \leq 2\wedge 1\leq{\tt x} \leq 1\).  It is contradictory
and cannot be satisfied when executing from valuations satisfying the
pre-condition, which is not contradictory.

To cope with this difficulty, we need to express that the abstract
interpreter works correctly only with
states that contain no duplicates.
We formalize this with a predicate {\tt consistent}, which is
defined as follows:
\begin{verbatim}
Fixpoint mem (s:string)(l:list string): bool :=
 match l with
   nil => false
 | x::l => if string_dec x s then true else mem s l
 end.

Fixpoint no_dups (s:state)(l:list string) :bool :=
 match s with
   nil => true
 | (s,_)::tl => if mem s l then false else no_dups tl (s::l)
 end.

Definition consistent (s:state) := no_dups s nil = true.
\end{verbatim}
The function {\tt no\_dups} actually returns {\tt true} when
the state {\tt s} contains no duplicates and no element from the 
exclusion list {\tt l}.  We prove, by induction on the of structure of
{\tt s}, that updating a state that satisfies {\tt no\_dups} for
an exclusion list {\tt l}, using {\tt a\_upd} for a variable {\tt x}
outside the exclusion list returns a new state that
still satisfies {\tt no\_dups} for {\tt l}.
The statement is as follows:
\label{th:no_dups_update}
\begin{verbatim}
Lemma no_dups_update :
  forall s l x v, mem x l = false ->
   no_dups s l = true -> no_dups (a_upd x v s) l = true.
\end{verbatim}
The proof of this lemma is done by induction on {\tt s}, making sure
that the property that is established for every {\tt s} is universally
quantified over {\tt l}: the induction hypothesis is actually used for
a different value of the the exclusion list.

The corollary from this lemma corresponding to the case where
{\tt l} is instantiated with the empty list
expresses that {\tt a\_upd} preserves the {\tt consistent} property.
\begin{verbatim}
Lemma consistent_update :
  forall s x v, consistent s -> consistent (a_upd x v s).
\end{verbatim}
\subsection{Proving the correctness of this interpreter}
When the interpreter runs on an instruction \(i\) and a state 
\(s\) and returns an annotated instruction \(i'\) and a new state \(s'\),
the correctness of the run is expressed with three properties:
\begin{itemize}
\item The assertion {\tt s\_to\_a \(s\)} is stronger than the
pre-condition\\ {\tt pc \(i'\) (s\_to\_a \(s'\))},
\item All the verification conditions in {\tt vc \(i'\) (s\_to\_a \(s'\))}
are valid,
\item The annotated instruction \(i'\) is an annotated version of the
input \(i\).
\end{itemize}
In the next few sections, we will prove that all runs of the abstract
 interpreter are correct.
\subsection{Soundness of abstract evaluation for expressions}
 When an expression \(e\) evaluates abstractly to an abstract value \(a\)
and concretely to an integer \(z\), \(z\) should satisfy
the predicate associated to the value \(a\).  Of course, the evaluation of
\(e\) can only be done using a valuation that takes care of providing values
for all variables occuring in \(e\).  This valuation must be consistent with
the abstract state that is used for the abstract evaluation leading to \(a\).
The fact that a valuation is consistent with an abstract state is simply
expressed by saying that the interpretation of the corresponding assertion
for this valuation has to hold.  Thus, the soundness of abstract evaluation
is expressed with a lemma that has the following shape:
\begin{verbatim}
Lemma a_af_sound : 
  forall s g e, ia m g (s_to_a s) ->
    ia m g (to_pred (a_af s e) (anum (af g e))).
\end{verbatim}
This lemma is proved by induction on the expression {\tt e}.  The case
where {\tt e} is a number is a direct application of the hypothesis
{\tt from\_Z\_sem}, the case where {\tt e} is an addition is a consequence
of {\tt a\_add\_sem}, combined with induction hypotheses.  The case
where {\tt e} is a variable relies on another lemma:
\begin{verbatim}
Lemma lookup_sem : forall s g, ia m g (s_to_a s) -> 
  forall x, ia m g (to_pred (lookup s x) (anum (g x))).
\end{verbatim}
This other lemma is proved by induction on {\tt s}.  In the base case,
{\tt s} is empty, {\tt lookup s x} is {\tt top}, and the hypothesis
{\tt top\_sem} makes it possible to conclude; in
the step case, if {\tt s} is {\tt (y,v)::s'} then the hypothesis 
\begin{verbatim}
ia m g (s_to_a s)
\end{verbatim}
reduces to
\begin{verbatim}
to_pred v (avar y) /\ ia m g (s_to_a s')
\end{verbatim}
We reason by cases on whether {\tt x} is {\tt y} or not.  If {\tt x} is
equal to {\tt y} then {\tt to\_pred v (avar y)} is the same
as {\tt to\_pred v (anum (g x))} according to
{\tt to\_pred\_sem} and {\tt lookup s x} is the same as {\tt v} by definition
of {\tt lookup}, this is enough to conclude this case.  If {\tt x} and
{\tt y} are different, we use the induction hypothesis on {\tt s'}.
\subsection{Soundness of update}
In the weakest pre-condition calculus, assignments of the
form {\tt x := e} are taken care of by
substituting all occurrences of the assigned variable {\tt x}
with the arithmetic expression {\tt e}
in the post-condition to obtain the weakest pre-condition.  In
the abstract interpreter, assignment is taken care of by updating the
first instance of the variable in the state.  There is a discrepancy between
the two approaches, where the first approach
 acts on all instances of the variable and the
second approach acts only on the first one.  This discrepancy is resolved in the
conditions of our experiment, where we work with abstract states that
contain only one binding for each variable: in this case, updating the
first variable is the same as updating all variables.  We express this
with the following lemmas:
\begin{verbatim}
Lemma subst_no_occur :
  forall s x l e,
   no_dups s (x::l) = true -> subst x e (s_to_a s) = (s_to_a s).
\end{verbatim}
\begin{verbatim}
Lemma subst_consistent :
   forall s g v x e, consistent s -> ia m g (s_to_a s) ->
     ia m g (to_pred v (anum (af g e))) ->
     ia m g (subst x e (s_to_a (a_upd x v s))).
\end{verbatim}
Both lemmas are proved by induction on {\tt s} and the second one uses
the first in the case where the substituted variable {\tt x} is
the first variable occuring in {\tt s}.  This proof also relies
on the hypothesis {\tt subst\_to\_pred}.

\subsection{Relating input abstract states and pre-conditions}
For the correctness proof we consider runs starting from
an instruction {\tt i} and an initial
 abstract state {\tt s} and obtaining an annotated instruction {\tt i'} and
a final abstract state {\tt s'}.  We
are then concerned with the verification conditions and the pre-condition
generated for the post-condition corresponding to {\tt s'} and the annotated
instruction {\tt i'}.  The pre-condition we obtain is either the assertion
corresponding to {\tt s} or the assertion {\tt a\_true}, when the first
sub-instruction in {\tt i} is a while loop.  In all cases,
the assertion corresponding to {\tt s} is stronger than the pre-condition.
This is expressed with the following lemma, which is easily proved by
induction on {\tt i}.
\begin{verbatim}
Lemma ab1_pc :
  forall  i i' s s', ab1 i s = (i', s') ->
    forall g a, ia m g (s_to_a s) -> ia m g (pc i' a).
\end{verbatim}
This lemma is actually stronger than needed, because the post-condition
used for computing the pre-condition does not matter, since the resulting
annotated instruction is heavily annotated with assertions
and the pre-condition always comes
from one of the annoations.

\subsection{Validity of generated conditions}
The main correctness statement only concerns states that satisfy
the {\tt consistent} predicate, that is, states that contain at most one
entry for each variable.  The statement is proved by induction on instructions.
As is often the case, what we prove by induction is a stronger statement;
Such a stronger statement
also means stronger induction hypotheses.  Here we add the information that
the resulting state is also consistent.

\begin{theorem} If \(s\) is a consistent state and running
the abstract interpreter {\tt ab1} on \(i\)
 from \(s\) returns a new annotated instruction \(i'\) and afinal state \(s'\),
then all the verification conditions
generated for \(i'\) and the post-condition associated to \(s'\) are
valid.  Moreover, the state \(s'\) is consistent.
\end{theorem}

The Coq encoding of this theorem is as follows:
\begin{verbatim}
Theorem ab1_correct : forall i i' s s', 
  consistent s -> ab1 i s = (i', s') ->
   valid m (vc i' (s_to_a s')) /\ consistent s'.
\end{verbatim}
This statement is proved by induction on {\tt i}.  Three cases arise,
corresponding to the three instructions in the language.

\begin{enumerate}
\item When {\tt i} is an assignment {\tt x := e}, this is the base case.
{\tt ab1 i s} computes to 
\begin{verbatim}
(pre (s_to_a s) (a_assign x e), a_upd x (a_af s e) s)
\end{verbatim}
From the lemma {\tt a\_af\_sound} we obtain that the concrete value
of {\tt e} in any valuation {\tt g} that satisfies {\tt ia m g (s\_to\_a s)}
satisfies the following property:
\begin{verbatim}
ia m g (to_pred (a_af s e) (anum (af g e)))
\end{verbatim}
The lemma {\tt subst\_consistent} can then be used to obtain the validity
of the following condition.
\begin{verbatim}
imp (s_to_a s) (subst x e (s_to_a (a_upd x (a_af s e) s)))
\end{verbatim}
This is the single verification condition generated for this instruction.
The
second part is taken care of by {\tt consistent\_update}.
\item When the instruction \hbox{\tt i} is a sequence \hbox{\tt seq i1 i2},
the abstract interpreter first processes \hbox{\tt i1} with the state \hbox{\tt s} as
input to obtain an annotated instruction \hbox{\tt a\_i1} and
an output state \hbox{\tt s1},  it then processes \hbox{\tt i2} with
\hbox{\tt s1} as input to obtain an annotated
instruction \hbox{\tt a\_i2} and a state \hbox{\tt s2}. 
The state \hbox{\tt s2} is used as the output state for the
whole instruction.  We then need to verify that the conditions generated
for \hbox{\tt a\_seq a\_i1 a\_i2} using \hbox{\tt s\_to\_a a2} as post-condition are
valid and \hbox{\tt s2} satisfies the \hbox{\tt consistent} property.
The conditions can be split in two parts.  The second part is
\hbox{\tt vc a\_i2 (s\_to\_a a2)}.  the validity of these conditions is a direct
consequence of the induction hypotheses.  The first part is
\hbox{\tt vc a\_i1 (pc a\_i2 (s\_to\_a s2))}.  This is not a direct consequence
of the induction hypothesis, which only states \hbox{\tt vc a\_i1 (s\_to\_a s1)}.
However, the lemma \hbox{\tt ab1\_pc} applied on \hbox{\tt a\_i2} states that
\hbox{\tt s\_to\_a s1} is stronger than \hbox{\tt pc (s\_to\_a s2)} and the lemma
\hbox{\tt vc\_monotonic} makes it possible to conclude.   With respect to
the \hbox{\tt consistent} property, it is recursively transmitted from \hbox{\tt s}
to \hbox{\tt s1} and from \hbox{\tt s1} to \hbox{\tt s2}.
\item When the instruction is a while loop, the body of the loop is recursively
processed with the \hbox{\tt nil} state, which is always satisfied.  Thus, the
verification conditions all conclude to \hbox{\tt a\_true} which is trivially
true.  Also, the \hbox{\tt nil} state also trivially satisfies the
\hbox{\tt consistent} property.
\end{enumerate}
\subsection{The annotated instruction}
We also need to prove that the produced annotated instruction
really is an annotated version of the initial instruction.
To state this new lemma, we first define a simple function that forgets the
annotations in an annotated instruction:
\begin{verbatim}
Fixpoint cleanup (i: a_instr) : instr :=
  match i with
    pre a i => cleanup i
  | a_assign x e => assign x e
  | a_seq i1 i2 => seq (cleanup i1) (cleanup i2)
  | a_while b a i => while b (cleanup i)
  end.
\end{verbatim}
We then prove a simple lemma about the abstract interpreter and
this function.
\begin{verbatim}
Theorem ab1_clean : forall i i' s s',
   ab1 i s = (i', s') -> cleanup i' = i.
\end{verbatim}
The proof of this lemma is done by induction on the structure of \hbox{\tt i}.
\subsection{Instantiating the simple abstract interpreter}
We can instantiate this simple abstract interpreter on a data-type of
odd-even values, using the following inductive type and functions:
\begin{verbatim}
Inductive oe : Type := even | odd | oe_top.

Definition oe_from_Z (n:Z) : oe :=
  if Z_eq_dec (Zmod n 2) 0 then even else odd.

Definition oe_add (v1 v2:oe) : oe :=
  match v1,v2 with
    odd, odd => even
  | even, even => even
  | odd, even => odd
  | even, odd => odd
  | _, _ => oe_top
  end.
\end{verbatim}
The abstract values can then be mapped into assertions in the obvious way
using a function {\tt oe\_pred} which we do not describe here for the
sake of conciseness.
Running this simple interpreter on a small example, representing the program
\begin{verbatim}
 x := x + y; y := y + 1
\end{verbatim}
for the state {\tt ("x", odd)::("y", even)::nil} is represented by
the following dialog:
\begin{alltt}
Definition ab1oe := ab1 oe oe_from_Z oe_top oe_add oe_to_pred.

Eval vm_compute in
 ab1oe (seq (assign "x" (aplus (avar "x") (avar "y")))
            (assign "y" (aplus (avar "y") (anum 1))))
  (("x",even)::("y",odd)::nil).
\textit{    = (a_seq
        (pre
          (a_conj (pred "even" (avar "x" :: nil))
             (a_conj (pred "odd" (avar "y" :: nil)) a_true))
          (a_assign "x" (aplus (avar "x") (avar "y"))))
        (pre
          (a_conj (pred "odd" (avar "x" :: nil))
             (a_conj (pred "odd" (avar "y" :: nil)) a_true))
          (a_assign "y" (aplus (avar "y") (anum 1)))),
       ("x", odd) :: ("y", even) :: nil)
     : a_instr * state oe}
\end{alltt}
\section{A stronger interpreter}
More precise results can be obtained for while loops. For each loop
we need to find a state whose interpretation as an assertion will be an
acceptable invariant for the loop.  We want this invariant to take into
account any information that can be extracted from the boolean test in the
loop: when entering inside the loop, we know that the
test succeeded; when exiting the loop we know that the
test failed.  It turns out that this information can help us detect cases
where the body of a loop is never executed and cases where a loop can
never terminate.  To describe non-termination, we change the type of
values returned by the abstract interpreter: instead of returning an annotated
instruction and a state, our new abstract interpreter returns an
annotated instruction and an optional state: the optional value is {\tt None}
when we have detected that execution cannot terminate.  This detection
of guaranteed non-termination is conservative: when the analyser cannot
guarantee that an instruction loops, it returns a state as usual.  The presence
of optional states will slightly complexify the structure of our
static analysis.

We assume the existence of two new functions for this purpose.
\begin{itemize}
\item {\tt learn\_from\_success : state -> bexpr -> option state}, this is used
to encode the information learned when the test succeeded.
For instance if the environment initially contains
an interval {\tt [0,10]} for the variable {\tt x} and the test is {\tt x < 6},
then we can return the environment so that the value for {\tt x} becomes
{\tt [0, 5]}.  Sometimes, the initial environment is so that the test can
never be satisfied, in this case a value {\tt None} is returned instead of
an environment.
\item {\tt learn\_from\_failure : state -> bexpr -> option state},
this is used to compute
information about a state knowing that a test failed.
\end{itemize}
The body of a while loop is often meant to be run several times.  In
abstract interpretation, this is also true.  At every run, the information
about each variable at each location of the instruction needs to be updated to
take into account more and more concrete values that may be reached at this
location.  In traditional approaches to abstract interpretation, a binary
operation is applied at each location, to combine the information previously
known at this location and the new values discovered in the current run.
This is modeled by a binary operation.
\begin{itemize}
\item {\tt join : A -> A -> A}, this function takes two abstract values
and returns a new abstract value whose interpretation as a set is
larger than the two inputs.
\end{itemize}
The theoretical description of abstract interpretation insists that
the set {\tt A}, together with the values {\tt join} and {\tt top}
should constitute an upper semi-lattice.  In fact, We will use only part of
the properties of such a structure in our proofs about the abstract
interpreter.

When the functions {\tt learn\_from\_success} and {\tt learn\_from\_failure}
return a {\tt None} value, we actually detect that some code will never
be executed.  For instance, if {\tt learn\_from\_success} returns {\tt None},
we can know that the test at the entry of a loop will never be satisfied
and we can conclude that the body of the loop is not executed.  In this
condition, we can mark this loop body with a false assertion.  We provide a
function for this purpose:
\begin{verbatim}
Fixpoint mark (i:instr) : a_instr :=
  match i with
    assign x e => pre a_false (a_assign x e)
  | seq i1 i2 => a_seq (mark i1) (mark i2)
  | while b i => a_while b a_false (mark i)
  end.
\end{verbatim}
Because it marks almost every instruction, this function makes it easy
to recognize at first glance the fragments of code that are dead code.
A more lightweight approach could be to mark only the sub-instructions
for which an annotation is mandatory: while loops.

\subsection{Main structure of invariant search}
\label{inv-search-structure}
In general, finding the most precise invariant for a while loop is
an undecidable problem.  Here we are describing a static analysis tool.
We will trade preciseness for guaranteed termination.  The approach
we will describe will be as follows:
\begin{enumerate}
\item Run the body of the loop abstractly for a few times, progressively
widening the sets of values for each variable at each run.  If this process
stabilizes, we have reached an invariant,
\item If no invariant was reached, try taking over-approximations of the
values for some variables and run again the loop for a few times.  This
process may also reach an invariant,
\item If no invariant was reached by progressive widening, pick an abstract
state that is guaranteed to be an invariant (as we did for the first
simple interpreter: take the {\sl top} state that gives no information about any
variable),
\item Invariants that were obtained by over-approximation
can then be improved by a {\em narrowing} process:
when run through the loop again, even if no information about the state
is given at the beginning of the loop, we may still be able to gather
some information at the end of executing the loop.  The state
that gathers the information at the end of the loop and the information
before entering the loop is most likely to be an invariant, which is
more precise (narrower) than the top state.  Again this process may be
run several times.
\end{enumerate}
We shall now review the operations involved in each of these steps.
\subsection{Joining states together}
Abstract states are finite list of pairs of variable names and
abstract values.  When a variable does not occur in a state, the
associated abstract value is {\tt top}.  When joining two states together
every variable that does not occur in one of the two states should receive
the {\tt top} value, and every variable that occurs in both states should
receive the join of the two values found in each state.
We describe this by writing a function that studies all the variables
that occur in one of the lists: it is guaranteed to perform the right behavior
for all the variables in both lists, it naturally associates the {\tt top}
value to the variables that do not occur in the first list (because no pair
is added for these variables), and it naturally associates the {\tt top}
value to the variables that do not occur in the second list, because
{\tt top} is the value found in the second list and {\tt join} preserves
{\tt top}.
\begin{verbatim}
Fixpoint join_state (s1 s2:state) : state :=
  match s1 with
    nil => nil
  | (x,v)::tl => a_upd x (join v (lookup s2 x)) (join_state tl s2)
  end.
\end{verbatim}
Because we sometimes detect that some instruction will not be executed
we occasionally have to consider situation were we are not given a state
after executing a while loop.  In this case, we have to combine together
a state and the absence of a state.  But because the absence of state
corresponds to a false assertion, the other state is enough to describe
the required invariant.  We encode this in an auxiliary function.
\begin{verbatim}
Definition join_state' (s: state)(s':option state) : state :=
   match s' with
     Some s' => join_state s s'
   | None => s
   end.
\end{verbatim}
\subsection{Running the body a few times}
In our general description of the abstract interpretation of loops, we
need to execute the body of loops in two different modes: one mode is
a {\em widening} mode the other is a {\em narrowing} mode.
In the narrowing mode, after executing the body
of the loop needs to be joined with the initial state before executing
the body of the loop, so that the result state is less precise than
both the state before executing the body of the loop and the
state after executing
the body of the loop.  In the {\em narrowing} mode, we start the execution with
an environment that is guaranteed to be large enough, hoping
to narrow this environment to a more precise value.
In this case, the {\tt join} operation must not be done with the state
that is used to start the execution, but with another state which
describes the information known about variables before considering the loop.
To accomodate these two modes of abstract execution, we use a function
that takes two states as input: the first state is the one with which the
result must be joined, the second state is the one with which execution
must start.  In this function, the argument {\tt ab} is the function
that describes the abstract interpretation on the instruction inside the
loop, the argument {\tt b} is the test of the loop.  The function {\tt ab}
returns an optional state and an annotated instruction.  The optional
state is {\tt None} when the abstract interpreter can detect that
the execution of the program from the input state will never terminate.
When putting all elements together, the argument {\tt ab} will be instantiated
with the recursive call of the abstract interpreter on the loop body.
\begin{verbatim}
Definition step1 (ab: state -> a_instr * option state)
  (b:bexpr) (init s:state) : state :=
  match learn_from_success s b with
    Some s1 => let (_, s2) := ab s1 in join_state' init s2
  | None => s
  end.
\end{verbatim}
We then construct a function that repeats {\tt step1} a certain
number of times.  This number is denoted by a natural number {\tt n}.
In this function, the constant 0 is a natural number and we need to
make it precise to Coq's parser, by expressing that the value must
be interpreted in a parsing scope for natural numbers instead of
integers, using the specifier {\tt \%nat}.
\begin{verbatim}
Fixpoint step2 (ab: state ->  a_instr * option state)
 (b:bexpr) (init s:state) (n:nat) : state :=
 match n with
   0%nat => s
 | S p => step2 ab b init (step1 ab b init s) p
 end.
\end{verbatim}
The complexity of these functions can be improved: there is no
need to compute all iterations if we can detect early that a fixed point
was reached.  In this paper, we prefer to keep the code of the
abstract interpreter simple but potentially inefficient to make our
formal verification work easier.

\subsection{Verifying that a state is more precise than another}
To verify that we have reached an invariant, we need to check for a state
{\tt s}, so that running this state through {\tt step1 ab b s s} returns
a new state that is not less precise than {\tt s}.  For this, we assume
that there exist a function that makes it possible to compare two
abstract values:
\begin{itemize}
\item {\tt thinner : A -> A -> bool}, this function returns {\tt true}
when the first abstract value gives more precise information than the
second one.
\end{itemize}
Using this basic function on abstract values, we define a new function
on states:
\begin{verbatim}
Fixpoint s_stable (s1 s2 : state) : bool :=
  match s1 with
    nil => true
  | (x,v)::tl => thinner (lookup s2 x) v && s_stable tl s2
  end.
\end{verbatim}
This function traverses the first state to check that the abstract
value associated to each variable is less precise than the information
found in the second state.  This function is then easily used to
verify that a given state is an invariant through the abstract interpretation
of a loop's test and body.
\begin{verbatim}
Definition is_inv (ab:state-> a_instr * option state)
  (s:state)(b:bexpr):bool := s_stable s (step1 ab b s s).
\end{verbatim}
\subsection{Narrowing a state}
The {\tt step2} function receives two arguments of type {\tt state}.  The first
argument is solely used for join operations, while the second argument
is used to start a sequence of abstract states that correspond to iterated
interpretations of the loop test and body.  When the start state is not
stable through interpretation, the resulting state is larger than both
the first argument and the start argument.  When the start state is stable
through interpretation, there are cases where the resulting state is
smaller than the start state.

For instance, in the cases where the abstract values are {\tt even}
and {\tt odd}, if the first state argument maps the variable {\tt y} to
{\tt even} and the variable {\tt z} to {\tt odd}, the start state maps
{\tt y} and {\tt z} to the top abstract value
(the abstract value that gives no information) and the while loop is the
following:
\begin{verbatim}
while (x < 10) do x := x + 1; z:= y + 1; y := 2 done
\end{verbatim}
Then, after abstractly executing the loop test and body once, we obtain a
state where {\tt y} has the value {\tt even} and {\tt z} has the top 
abstract value.  This state is more precise
than the start state.  After abstractly executing the loop test and body
a second time, we obtain a state where {\tt z} has the value {\tt odd}
and {\tt y} has the value {\tt even}.  This state is more precise than the one
obtained only after the first abstract run of the loop test and body.

The example above shows that over-approximations are improved by running
the abstract interpreter again on them.  This phenomenon is known as
{\em narrowing}.  It is worth forcing a narrowing phase after each phase
that is likely to produce an over-approximation of the smallest fixed-point
of the abstract interpreter.  This is used in the abstract interpreter
that we describe below.

\subsection{Allowing for over-approximations}
In general, the finite amount of abstract computation performed in 
the {\tt step2} function is not enough to reach the smallest 
stable abstract state.  This is related to the undecidability of the
halting problem: it is often possible to write a program where a variable
will receive a precise value exactly when some other program terminates.
If we were able to compute the abstract value for this variable in a finite
amount of time, we would be able to design a program that solves the halting
problem.

Even if we are facing a program where finding the smallest state can be
done in a finite amount of time, we may want to accelerate the process by
taking over-approximations.  For instance, if we consider the following loop:
\begin{verbatim}
while x < 10 do x := x + 1 done
\end{verbatim}
If the abstract values we are working with are intervals and we start with
the interval {\tt [0,0]}, after abstractly interpreting the loop test and
body once, we obtain that the value for {\tt x} should contain at least
{\tt [0,1]}, after abstractly interpreting 9 times, we obtain that the
value for {\tt x} should contain at least {\tt [0,9]}.  Until these 9
executions, we have not seen a stable state.  At the 10th execution, we
obtain that the value for {\tt x} should contain at least {\tt [0, 10]} and
the 11th execution shows that this value actually is stable.

At any time before a stable state is reached, we may choose to replace the
 current
unstable state with a state that is ``larger''.  For instance, we may
choose to replace {\tt [0,3]} with {\tt [0,100]}.  When this happens, the
abstract interpreter can discover that the resulting state after starting with
the one that maps {\tt x} to {\tt [0,100]} actually maps {\tt x} to
 {\tt [0,10]}, thus {\tt [0,100]} is stable and is good candidate to enter
a narrowing phase.  This narrowing phase actually converges to a state
that maps {\tt x} to {\tt [0,10]}.

The choice of over-approximations is arbitrary and information may actually
be lost in the process, because over-approximated states are less precise,
 but this is compensated by the fact that the abstract
interpreter gives quicker answers.  The termination of the abstract interpreter
can even be guaranteed if we impose that a guaranteed over-approximation
is taken after a finite amount of steps.  An example of a guaranteed
over-approximation is a state that maps every variable to the top abstract
value.  In our Coq encoding, such a state is represented by the
{\tt nil} value.

The choice of over-approximation strategies varies from one abstract domain
to the other.  In our Coq encoding, we chose to let this over-approximation
be represented by a function with the following signature:
\begin{itemize}
\item {\tt over\_approx : nat -> state -> state -> state}  When applied
to {\tt n}, {\tt s}, and {\tt s'}, this function computes an over approximation
of {\tt s'}.  The value {\tt s} is supposed to be a value that comes before
{\tt s'} in the abstract interpretation and can be used to choose the
over-approximation cleverly, as it gives a sense of direction to the
current evolution of successive abstract values.  The number {\tt n} should
be used to fine-tune the coarseness of the over-approximation: the lower
the value of {\tt n}, the coarser the approximation.
\end{itemize}
For instance, when considering the example above, knowing that
\(s={\tt [0,1]}\) and \(s'={\tt [0,2]}\) are
two successive unstable values reached
by the abstract interpreter for the variable {\tt x} can suggest to choose an
over-approximation where the upper bound changes but the lower bound remains
unchanged.  In this case, we expect the function {\tt over\_approx} to
return {\tt [0,\(+\infty\)]}, for example.

\subsection{The main invariant searching function}
We can now describe the function that performs the process described
in section~\ref{inv-search-structure}.  The code of this function is
as follows:
\begin{verbatim}
Fixpoint find_inv ab b init s i n : state :=
  let s' := step2 ab b init s (choose_1 s i) in
  if is_inv ab s' b then s' else
    match n with 
      0%nat => nil
    | S p => find_inv ab b init (over_approx p s s') i p
    end.
\end{verbatim}
The function {\tt choose\_1} is provided at the same time as
all other functions that are specific to the abstract domain {\tt A}, such
as {\tt join}, {\tt a\_add}, etc.

The argument function {\tt ab} is supposed to be the function that performs
the abstract interpretation of the loop inner instruction {\tt i} (also
called the loop body),
the boolean expression {\tt b} is supposed to be the loop test.  The
state {\tt init} is supposed to be the initial input state at the first
invocation of {\tt find\_inv} on this loop and {\tt s} is supposed to be
the current over-approximation of {\tt init}, {\tt n} is the number
of over-approximations that are still allowed before the function
should switch to the {\tt nil} state, which is a guaranteed
over-approximation.  This function systematically runs the abstract
interpreter on the inner instruction an arbitrary number of times
(given by the function {\tt choose\_1}) and then tests whether the resulting
state is an invariant.  Narrowing steps actually take place if the number
of iterations given by {\tt choose\_1} is large enough.  If the result
of the iterations is an invariant, then it is returned.  When the result
state is not an invariant, the function {\tt find\_inv} is called recursively
with a larger approximation computed by {\tt over\_approx}.  When the
number of allowed recursive calls is reached, the {\tt nil} value is returned.

\subsection{Annotating the loop body with abstract information}
The {\tt find\_inv} function only produces a state, while the abstract
interpreter is also supposed to produce an annotated version of the
instruction.  Once we know the invariant, we can annotate the while loop
with this invariant and obtain an annotated version of the loop body
by re-running the abstract interpreter on this instruction.  This is
done with the following function:
\begin{verbatim}
Definition do_annot (ab:state-> a_instr * option state)
  (b:bexpr) (s:state) (i:instr) : a_instr :=
  match learn_from_success s b with
    Some s' => let (ai, _) := ab s' in ai
  | None => mark i
  end.
\end{verbatim}
In this function, {\tt ab} is supposed to compute the abstract interpretation
of the loop body.  When the function {\tt learn\_from\_success} returns
a {\tt None} value, this means that the loop body is never executed and
it is marked as dead code by the function {\tt mark}.
\subsection{The abstract interpreter's main function}
With the function {\tt find\_inv}, we can now design a new abstract
interpreter.  Its main structure is about the same as for the naive
one, but there are two important differences.  First, the abstract
interpreter now uses the {\tt find\_inv} function to compute an
invariant state for the while loop.  Second, this abstract interpreter
can detect cases where instructions are guaranteed to not terminate.
This is a second part of dead code detection: when a good invariant is
detected for the while loop, a comparison between this invariant and
the loop test may give the information that the loop test can never be
falsified.  If this is the case, no state is returned and the
instructions following this while loop in sequences must be marked as
dead code.  This is handled by the fact that the abstract interpreter
now returns an optional state and an annotated instruction.  The case
for the sequence is modified to make sure instruction are marked as
dead code when receiving no input state.
\begin{verbatim}
Fixpoint ab2 (i:instr)(s:state) : a_instr*option state :=
  match i with
    assign x e =>
    (pre (s_to_a s) (a_assign x e), Some (a_upd x (a_af s e) s))
  | seq i1 i2 => 
    let (a_i1, s1) := ab2 i1 s in
      match s1 with
        Some s1' =>
        let (a_i2, s2) := ab2 i2 s1' in
          (a_seq a_i1 a_i2, s2)
      | None => (a_seq a_i1 (mark i2), None)
      end
  | while b i =>
    let inv := find_inv (ab2 i) b s s i (choose_2 s i) in
        (a_while b (s_to_a inv)
              (do_annot (ab2 i) b inv i),
         learn_from_failure inv b) 
  end.
\end{verbatim}
This function relies on an extra numeric function {\tt choose\_2} to decide
the number of times {\tt find\_inv} will attempt progressive
over-approximations before giving up and falling back on the {\tt nil}
state.  Like {\tt choose\_1} and {\tt over\_approx}, this function must
be provided at the same time as the type for abstract values.
\section{Proving the correctness of the abstract interpreter}
To prove the correctness of our abstract interpreter, we adapt the correctness
statements that we already used for the naive interpreter.  The main
change is that the resulting state is optional, with a {\tt None} value
corresponding to non-termination.  This means that when a {\tt None} value
is obtained we can take the post-condition as the false assertion.  This is
expressed with the following function, mapping an optional state to an
assertion.
\begin{verbatim}
Definition s_to_a' (s':option state) : assert :=
  match s' with Some s => s_to_a s | None => a_false end.
\end{verbatim}
The main correctness statement thus becomes the following one:
\begin{verbatim}
Theorem ab2_correct : forall i i' s s', consistent s ->
  ab2 i s = (i', s') -> valid m (vc i' (s_to_a' s')).
\end{verbatim}
By comparison with the similar theorem for {\tt ab1}, we removed
the part about the final state satisfying the {\tt consistent}.
This part is actually proved in a lemma beforehand.  The reason
why we chose to establish the two results at the same time for {\tt ab1}
and in two stages for {\tt ab2} is anecdotal.

As for the naive interpreter this theorem is paired with a lemma asserting
that cleaning up the resulting annotated instruction {\tt i'} yields back
the initial instruction {\tt i}.  We actually need to prove two lemmas,
one for the {\tt mark} function (used to mark code as dead code) and one
for {\tt ab2} itself.
\begin{verbatim}
Lemma mark_clean : forall i, cleanup (mark i) = i.
\end{verbatim}
\begin{verbatim}
Theorem ab2_clean : forall i i' s s',
   ab2 i s = (i', s') -> cleanup i' = i.
\end{verbatim}
These two lemmas are proved by induction on the structure of the instruction
{\tt i}.

\subsection{Hypotheses about the auxiliary functions}
The abstract interpreter relies on a collection of functions that are
specific to the abstract domain being handled.  In our Coq development,
this is handled by defining the function inside a section, where the
various components that are specific to the abstract domain of interpretation
are given as section variables and hypotheses.  When the section is closed,
the various functions defined in the section are abstracted over the variables
that they use.  Thus, the function {\tt ab2} becomes a 16-argument function.
The extra twelve arguments are as follows:
\begin{enumerate}
\item {\tt A : Type}, the type containing the abstract values,
\item {\tt from\_Z : Z -> A}, a function mapping integer values to abstract
values,
\item {\tt top : A}, an abstract value representing lack of information,
\item {\tt a\_add : A -> A -> A}, an addition operation for abstract values,
\item {\tt to\_pred : A -> aexpr -> assert}, a function mapping abstract
values to their interpretations as assertions on arithmetic expressions,
\item {\tt learn\_from\_success : state A -> bexpr -> state A}, a function
that is able to improve a state, knowing that a boolean expression's evaluation
returns {\tt true},
\item {\tt learn\_from\_failure : state A -> bexpr -> state A}, similar
to the previous one, but using the knowledge that the boolean expression's
evaluation returns {\tt false},
\item {\tt join : A -> A -> A}, a binary function on abstract values that
returns an abstract value that is coarser than the two inputs,
\item {\tt thinner : A -> A -> bool}, a comparison function that succeeds
when the first argument is more precise than the second,
\item {\tt over\_approx : nat -> state A -> state A -> state A}, a function
that implements heuristics to find over-approximations of its arguments,
\item {\tt choose\_1 : state A -> instr -> nat}, a function that returns
the number of times a loop body should be executed with a given start state
before testing for stabilisation,
\item {\tt choose\_2 : state A -> instr -> nat}, a function that returns
the number of times over-approximations should be attempted before giving
up and using the coarsest state.
\end{enumerate}
Most of these functions must satisfy a collection of properties to ensure
that the correctness statement will be provable.  There are fourteen
such properties,
which can be sorted in the following way:
\begin{enumerate}
\item Three properties are concerned with the assertions
created by {\tt to\_pred}, with respect to their
logical interpretation and to substitution.
\item Two properties are concerned with the consistency of interpretation
of abstract values obtained through {\tt from\_Z} and {\tt a\_add} as predicates
over integers.
\item Two properties are concerned with the logical properties of abstract
states computed with the help of {\tt learn\_from\_success} and
{\tt learn\_from\_failure}.
\item Four properties are concerned with ensuring that {\tt over\_approx},
{\tt join}, and {\tt thinner} do return or detect over-approximations correctly,
\item Three properties are concerned with ensuring that the {\tt consistent}
properties is preserved through  {\tt learn\_from...} and {\tt over\_approx}.
\end{enumerate}

\subsection{Maintaining the {\tt consistent} property}
For this abstract interpreter, we need again to prove that it maintains
the property that all states are duplication-free.  It is first established
for the {\tt join\_state} operation.  Actually, since the {\tt join\_state} operation
performs repetitive updates from the {\tt nil} state, the result
is duplication-free, regardless of the duplications in the
inputs.  This is easily obtained with a proof by induction on the first
argument.  For once, we show the full proof script.
\begin{verbatim}
Lemma join_state_consistent :
  forall s1 s2, consistent (join_state s1 s2).
intros s1 s2; induction s1 as [ | [x v] s1 IHs1]; simpl; auto.
apply consistent_update; auto.
Qed.
\end{verbatim}
The first two lines of this Coq excerpt give the theorem statement.
The line \hbox{\tt intros ...} explains that a proof by induction
should be done.  This proof raises two cases, and the \hbox{as ...}
fragment states that in the step case (the second case),
one should consider a list whose
tail is named {\tt s1} and whose first pair contains a variable {\tt x}
and an abstract value {\tt v}, and we have an induction hypothesis, which
should be named {\tt IHs1}: this induction hypothesis states that {\tt s1}
already satisfies the {\tt consistent} property.
The {\tt simpl} directive expresses
that the recursive function should be simplified if possible, and {\tt auto}
attempts to solve the goals that are generated.  Actually, the
computation of recursive functions leads to proving {\tt true = true}
in the base case and {\tt auto} takes care of this.  For the step
case, we simply need to rely on the theorem {\tt consistent\_update}
(see section~\ref{th:no_dups_update}).  The premise of this theorem
actually is {\tt IHs1} and {\tt auto} finds it.
\subsection{Relating input abstract states and pre-conditions}
Similarly to what was done for the naive abstract interpreter, we
want to ensure that the interpretation of the input abstract
state as 
a logical formula implies the pre-condition for the generated annotated
instruction and the generated post-condition.  For the while loop, this
relies on the fact that the selected invariant is obtained after repetitive
joins with the input state. We first establish two monotonicity properties for
the {\tt join\_state} function, we show only the first one:
\begin{verbatim}
Lemma join_state_safe_1 : forall g s1 s2,
  ia m g (s_to_a s1) -> ia m g (s_to_a (join_state s1 s2)).
\end{verbatim}
Then, we only need to propagate the property
up from the {\tt step1} function.  Again, we show only the first one
but there are similar lemmas for {\tt step2}, {\tt find\_inv}; and
we conclude with the property for {\tt ab2}:
\begin{verbatim}
Lemma step1_pc : forall g ab b s s', 
  ia m g (s_to_a s) -> ia m g (s_to_a s') -> 
  ia m g (s_to_a (step1 ab b s s')).
\end{verbatim}
\begin{verbatim}
Lemma ab2_pc :
  forall  i i' s s', ab2 i s = (i', s') ->
    forall g a, ia m g (s_to_a s) -> ia m g (pc i' a).
\end{verbatim}
The proof for {\tt step1\_pc} is a direct consequence of the definition
and the properties of {\tt join\_state}.  The proofs for {\tt step2} and
{\tt find\_inv} are done by induction on {\tt n}.  The proof for
{\tt ab2} is an easy induction on the instruction {\tt i}.  In particular,
the two state arguments to the function {\tt find\_inv} are both equal
to the input state in the case of {\tt while} loops.
\subsection{Validity of the generated conditions}
The main theorem is about ensuring that all verification conditions
are provable.  A good half of this problem is already taken care of when
we prove the theorem {\tt ab2\_pc}, which expresses that at each step
the state is strong enough to ensure the validity of the pre-condition for
the instruction that follows.  The main added difficulty is to verify
that the invariant computed for each while loop actually is invariant.
This difficulty is taken care of by the structure of the function
{\tt find\_inv}, which actually invokes the function {\tt is\_inv} on its
expected output before returning it.  Thus, we only need to prove that
{\tt is\_inv} correctly detects states that are invariants:
\begin{verbatim}
Lemma is_inv_correct :
   forall ab b g s s' s2 ai,
     is_inv ab s b = true -> learn_from_success s b = Some s' ->
     ab s' = (ai, s2) -> ia m g (s_to_a' s2) -> ia m g (s_to_a s).
\end{verbatim}
We can then deduce that {\tt find\_inv} is correct: the proof proceeds
by showing that the value this function returns is either verified using
{\tt is\_inv} or the {\tt nil} state.  The correctness statement for
{\tt find\_inv} has the following form:
\begin{verbatim}
Lemma find_inv_correct : forall ab b g i n init s s' s2 ai,
  learn_from_success (find_inv ab b init s i n) b = Some s' ->
  ab s' = (s2, ai) -> ia m g (s_to_a' s2) ->
  ia m g (s_to_a (find_inv ab b init s i n)).
\end{verbatim}
This can then be combined with the assumptions that
{\tt learn\_from\_success} and {\tt learn\_from\_failure} correctly
improve the information given in abstract state to show that the
value returned for while loops in {\tt ab2} is correct.  These assumptions
have the following form (the hypothesis for the {\tt learn\_from\_failure}
has a negated third assumption).
\begin{verbatim}
Hypothesis learn_from_success_sem :
 forall s b g, consistent s -> 
      ia m g (s_to_a s) -> ia m g (a_b b) ->
      ia m g (s_to_a' (learn_from_success s b)).
\end{verbatim}

\section{An interval-based instantiation}
The abstract interpreters we have described so far are generic and
are ready to be instantiated on specific abstract domains.  In this
section we describe an instantiation on an abstract domain to
represent intervals.  This domain of intervals contains intervals with
finite bounds and intervals with infinite bounds.  The interval
with two infinite bounds represents the whole type of integers.  We
describe these intervals with an inductive type that has four variants:
\begin{verbatim}
Inductive interval : Type :=
  above : Z -> interval
| below : Z -> interval
| between : Z -> Z -> interval
| all_Z : interval.
\end{verbatim}
For instance, the interval containing all values larger than
or equal to 10 is represented by {\tt above 10} and the whole
type of integers is represented by {\tt all\_Z}.

The interval associated to an integer is simply described as
the interval with two finite bounds equal to this integer.
\begin{verbatim}
Definition i_from_Z (x:Z) := between x x.
\end{verbatim}

When adding two intervals, it suffices to add the two bounds,
because addition preserves the order on integers.
Coping with all the variants of
each possible input yields a function with many cases.
\begin{verbatim}
Definition i_add (x y:interval) :=
  match x, y with 
    above x, above y => above (x+y)
  | above x, between y z => above (x+y)
  | below x, below y => below (x+y)
  | below x, between y z => below (x+z)
  | between x y, above z => above (x+z)
  | between x y, below z => below (y+z)
  | between x y, between z t => between (x+z) (y+t)
  | _, _ => all_Z
  end.
\end{verbatim}

The assertions associated to each abstract value can rely on
only one, as we can re-use the same comparison predicate
for almost all variants.  This is described in the
{\tt to\_pred} function.
\begin{verbatim}
Definition i_to_pred (x:interval) (e:aexpr) : assert :=
  match x with
    above a => pred "leq" (anum a::e::nil)
  | below a => pred "leq" (e::anum a::nil)
  | between a b => a_conj (pred "leq" (anum a::e::nil))
                      (pred "leq" (e::anum b::nil))
  | all_Z => a_true
  end.
\end{verbatim}
Of course, the meaning attached to the string {\tt "leq"} must
be correctly fixed in the corresponding instantiation for the {\tt m}
parameter:

\begin{verbatim}
Definition i_m (s : string) (l: list Z) : Prop :=
  if string_dec s "leq" then
    match l with x::y::nil => x <= y | _ => False end
  else False.
\end{verbatim}
\subsection{Learning from comparisons}
The functions {\tt i\_learn\_from\_success} and {\tt i\_learn\_from\_failure}
used when processing while loops can be made arbitrarily complex.  For
the sake of conciseness, we have only designed a pair of functions
that detect the case where the boolean test has the form {\tt x < e},
where {\tt e} is an arbitrary arithmetic expression.  In this case,
the function {\tt i\_learn\_from\_success} updates
only the value associated to {\tt x}: the initial interval
associated with {\tt x} is intersected with the interval of all values
that are less than the upper bound of the interval computed for {\tt e}.
An impossibility is detected when the lowest possible value for {\tt x}
is larger than or equal to the upper bound for {\tt e}.  Even this
simple strategy yields a function with many cases, of which we
show only the cases where both {\tt x} and {\tt e} have interval
values with finite bounds:
\begin{verbatim}
Definition i_learn_from_success s b :=
  match b with
    blt (avar x) e =>
      match a_af _ i_from_Z all_Z i_add s e,
            lookup _ all_Z s x with
      ...
      | between _ n, between m p =>
           if Z_le_dec n m then None else
           if Z_le_dec n p
           then Some (a_upd _ x (between m (n-1)) s)
           else Some s
    ...
      end
   | _ => Some s
   end.
\end{verbatim}
In the code of this function, the functions {\tt a\_af}, {\tt lookup},
and {\tt a\_upd} are parameterized by the functions from the datatype
of intervals that they use: {\tt i\_from\_Z}, {\tt all\_Z} and {\tt i\_add}
for {\tt a\_af}, {\tt all\_Z} for {\tt lookup}, etc.

The function {\tt i\_learn\_from\_failure} is designed similarly,
looking at upper bounds for {\tt x} and lower bounds for {\tt e} instead.
\subsection{Comparing and joining intervals}
The treatement of loops also requires a function to find upper bounds
of pairs of intervals and a function to compare two intervals.  These
functions are simply defined by pattern-matching on the kind of intervals
that are encountered and then comparing the upper and lower bounds.
\begin{verbatim}
Definition i_join (i1 i2:interval) : interval :=
  match i1, i2 with
    above x, above y =>
      if Z_le_dec x y then above x else above y
  ...
  | between x y, between z t =>
      let lower := if Z_le_dec x z then x else z in
      let upper := if Z_le_dec y t then t else y in
      between lower upper
  | _, _ => all_Z
  end.

Definition i_thinner (i1 i2:interval) : bool :=
  match i1, i2 with
    above x, above y => if Z_le_dec y x then true else false
  | above _, all_Z => true
  ...
  | between x _, above y => if Z_le_dec y x then true else false
  | between _ x, below y => if Z_le_dec x y then true else false
  | _, all_Z => true
  ...
  end.
\end{verbatim}
\subsection{Finding over-approximations}
When the interval associated to a variable does not stabilize,
an over-approxi\-mation must be found for this interval.  We implement
an approach where several steps of over-approxi\-mation can be taken
one after the other.  For intervals, finding over-approxi\-mations
can be done by pushing one of the bounds of each interval to infinity.
We use the fact that the generic abstract interpreter calls the
over-approxi\-mation with two values to choose the bound that should be
pushed to infinity: in a first round of over-approxi\-mation, only the
bound that does not appear to be stable is modified.  This strategy
is particularly well adapted for loops where one variable is increased
or decreased by a fixed amount at each execution of the loop's body.

The strategy is implemented in two functions, the first function
over-approxi\-mates an interval, the second function applies the
first to all the intervalles found in a state.
\begin{verbatim}
Definition open_interval (i1 i2:interval) : interval :=
  match i1, i2 with
    below x, below y => if Z_le_dec y x then i1 else all_Z
  | above x, above y => if Z_le_dec x y then i1 else all_Z
  | between x y, between z t =>
    if Z_le_dec x z then if Z_le_dec t y then i1 else above x
    else if Z_le_dec t y then below y else all_Z
  | _, _ => all_Z
  end.

Definition open_intervals (s s':state interval) : state interval :=
  map (fun p:string*interval => 
         let (x, v) := p in 
           (x, open_interval v (lookup _ all_Z s' x))) s.
\end{verbatim}
The result of {\tt open\_interval i1 i2} is expected
to be an over-approximation of {\tt i1}.  The second argument
{\tt i2} is only used to choose which of the bounds of {\tt i1}
should be modified.

The function {\tt i\_over\_approx} receives a numeric parameter to
indicate the strength of over-approximation that should be applied.
Here, there are only two strengths: at the first try (when the level is 
larger
than 0), the function applies {\tt open\_intervals}; at the second try,
it simply returns the {\tt nil} state, which corresponds to the
{\tt top} value in the domain of abstract states.
\begin{verbatim}
Definition i_over_approx n s s' := 
  match n with
    S _ => open_intervals s s'
  | _ => nil
  end.
\end{verbatim}
The abstract interpreter also requires two functions that compute
the number of attempts at each level of repetitive operation.  We define
these two functions as constant functions:
\begin{verbatim}
Definition i_choose_1 (s:state interval) (i:instr) := 2%nat.
Definition i_choose_2 (s:state interval) (i:instr) := 3%nat.
\end{verbatim}

Once the type {\tt interval} and the various functions are provided we
obtain an abstract interpreter for computing with intervals.
\begin{verbatim}
Definition abi :=
 ab2 interval i_from_Z all_Z i_add i_to_pred
   i_learn_from_success i_learn_from_failure
   i_join i_thinner i_over_approx i_choose_1 i_choose_2.
\end{verbatim}

We can already run this instantiated interpreter inside the
Coq system.  For instance, we can run the interpreter
on the instruction:
\begin{verbatim}
  while x < 10 do x := x + 1 done
\end{verbatim}
This gives the following dialog (where
the answer of the Coq system is written in italics):
\begin{alltt}
Eval vm_compute in
  abi (while (blt (avar "x") (anum 10))
          (assign "x" (aplus (avar "x") (anum 1))))
      (("X", between 0 0)::nil).
\textit{     = (a_while (blt (avar "x") (anum 10))
        (a_conj
          (a_conj (pred "leq" (anum 0 :: avar "x" :: nil))
             (pred "leq" (avar "x" :: anum 10 :: nil))) a_true)
        (pre
          (a_conj
            (a_conj (pred "leq" (anum 0 :: avar "x" :: nil))
              (pred "leq" (avar "x" :: anum 9 :: nil))) a_true)
          (a_assign "x" (aplus (avar "x") (anum 1)))),
       Some (("x", between 10 10) :: nil))
     : a_instr * option (state interval)}
\end{alltt}
\section{Conclusion}
This paper describes how the functional language present in a higher-order
theorem prover can be used to encode a tool to perform a static analysis
on an arbitrary programming language.  The example programming language
is chosen to be extremely simple, so that the example can be described
precisely in this tutorial paper.  The static analysis tool that we
described is inspired by the approach of abstract interpretation.  However
this work is not a comprehensive introduction to abstract interpretation,
nor does it cover all the aspects of encoding abstract interpretation
inside a theorem prover.  Better descriptions of abstract interpretation
and its formal study are given in
\cite{PichardieThesis,TCSAppSem:BessonJensenPichardie,FICS08:Pichardie}.

The experiment is performed with the Coq system.  More extensive studies
of programming languages using this system have been developed over the
last years.  In particular, experiments by the Compcert team show that
not only static analysis but also efficient compilation can be described
and proved correct
 \cite{Bertot-Gregoire-Leroy-05,Leroy-compcert-06,Blazy-Dargaye-Leroy-06}.
  Coq is also used extensively for the study
of functional programming languages, in particular to study the properties
of type systems and there are a few Coq-based solutions to the
general landmark objective known as POPLMark \cite{aydemir05mechanized}.

The abstract interpreter we describe here is inefficient in many respects:
when analysing the body of a loop, this loop needs to be executed abstractly
several times, the annotations computed each time are forgotten, and then
when an invariant is discovered, the whole process needs to be done again
to produce the annotated instruction.  A more efficient interpreter could
be designed where computed annotations are kept in memory long enough to
avoid recomputation when the invariant is found.  We did not design the
abstract interpreter with this optimisation, thinking that the sources
of inefficiency could be calculated away through systematic transformation
of programs, as studied in another paper in this volume.  The
abstract interpreter provided with the paper \cite{Bertot08a} contains
some of these optimisations.

An important remark is that program analyses can be much more efficient
when they consider the relations between several variables at a time, as
opposed to the experiment described here where the variables are considered
independently of each other.  More precise work where relations between
variables can be tracked is possible, on the condition that abstract
values are used to describe complete states, instead of single variables
as in \cite{Bertot-Gregoire-Leroy-05}, where
the result of the analysis is used as a basis for a compiler
optimisation known as {\em common subexpression elimination}.

We have concentrated on a very simple while language in this paper, for
didactical purposes.  However, abstract interpreters have been applied to
much more complete programming languages.  For instance, the Astree 
\cite{astree05} analyser
covers most of the C programming language.  On the other hand, the foundational
papers describe abstract interpretation in terms of analyses on
control flow graphs.  The idea of abstract
interpretation is general enough that it should be possible to apply it to
any form of programming language.

\bibliography{a}
\bibliographystyle{plain}
\end{document}